\documentclass[showpacs,superscriptaddress,prd,twocolumn]{revtex4}
\usepackage{amsmath,amssymb,amsfonts,amsthm}
\usepackage{wasysym}

\newtheorem{theorem}{Theorem}[section]

\newcommand{\Si}{\mathcal{R}}

\newcommand{\Sip}{\mathcal{R}'}

\newcommand{\ra}{\mathcal{R}_{SY}}

\newcommand{\ram}{\mathcal{R}_{OM}}

\newcommand{\dv}{dv}

\newcommand{\dvf}{dv_0}

\newcommand{\norm}{\lambda}

\newcommand{\dom}{U}

\begin{document}
\title{Inequality between size and  angular momentum for bodies}
\author{Sergio Dain}
\email[E-mail: ]{dain@famaf.unc.edu.ar}
\affiliation{Facultad de Matem\'atica, Astronom\'{\i}a y
  F\'{i}sica, FaMAF, Universidad
  Nacional de C\'ordoba,\\
  Instituto de F\'{\i}sica Enrique Gaviola, IFEG, CONICET,\\
  Ciudad Universitaria (5000) C\'ordoba, Argentina}

\begin{abstract}
  A universal inequality that bounds the angular momentum of a body by the
  square of its size is presented and heuristic physical arguments are given to
  support it. We prove a version of this inequality, as consequence of Einstein
  equations, for the case of rotating axially symmetric, constant density,
  bodies.  Finally, the physical relevance of this result is discussed.
\end{abstract}

\date{\today}

\pacs{04.40Dg. 04.70.Bw, 04.20.Ex, 02.40.Ky}
\maketitle

\emph{Introduction} --- Consider a rotating body $\dom$ with angular momentum
$J(\dom)$. Let $\Si(\dom)$ be a measure (with units of length) of the size
of the body.  A precise definition for the radius $\Si$ will be given later on,
for the present discussion it is enough to consider only the intuitive idea of
size: for example, if the body is a sphere in flat space then $\Si$ should
be proportional to the radius of the sphere.

We conjecture that there exists a universal inequality for all bodies of the
form
\begin{equation}
  \label{eq:22}
  \Si^2(\dom) \apprge  \frac{G}{c^3} |J(\dom)|,
\end{equation}
where $G$ is the gravitational constant and $c$ the speed of light. The symbol
$\apprge$ is intended as an order of magnitude, the precise universal
(i.e. independent of the body) constant will depend, of course, on the
definition of $\Si$.

The purpose of the first part of this article is to provide physical arguments
supporting the validity of this inequality.  In the second part we prove, as
consequence of Einstein field equations, theorem \ref{t:1}. This theorem
provides a precise version of the inequality \eqref{eq:22} valid for rotating,
axially symmetric, constant density, bodies.  Finally, we conclude with a
discussion of the physical relevance of this result.

\emph{Heuristic arguments}---
The arguments in support of the inequality (\ref{eq:22}) are based in the
following three physical principles:
 \begin{itemize}
\item[(i)] The speed of light $c$ is the maximum speed.

\item[(ii)] For bodies which are not contained in a black hole the following
  inequality holds
  \begin{equation}
    \label{eq:2}
    \Si(\dom) \apprge\frac{G}{c^2}  m(\dom),
  \end{equation}
  where $m(\dom)$ is the mass of the body.

\item[(iii)] The inequality (\ref{eq:22})   holds for black holes.
\end{itemize}
Let us discuss these assumptions. Item (i) is clear. Item (ii) is called the
\emph{trapped surface conjecture} \cite{Seifert79}.  Essentially, it says that
if the reverse inequality as in (\ref{eq:2}) holds then a trapped surface
should enclose $\dom$. That is: if matter is enclosed in a sufficiently small
region, then the system should collapse to a black hole.  This is related with
the \emph{hoop conjecture} \cite{thorne72} (see also \cite{Wald99}
\cite{PhysRevD.44.2409} \cite{Malec:1992ap} ).  The trapped surface
conjecture has been proved in spherical symmetry \cite{Bizon:1989xm}
\cite{Bizon:1988vv} \cite{Khuri:2009dt} and  also for  a relevant class
of non-spherical initial data \cite{Malec:1991nf}. The general case remains
open but it is expected that some version of this conjecture should hold.

Concerning item (iii),  the inequality
\begin{equation}
  \label{eq:5}
  A\geq8\pi\frac{G}{c^3} |J|
\end{equation}
was recently proved for axially symmetric black holes (see \cite{dain12} and
reference therein), where $A$ is the area of the stable marginally trapped surface and $J$ its
angular momentum.  The area $A$ is a measure of the size of a trapped surface,
hence the inequality \eqref{eq:5} represents a version of \eqref{eq:22} for
axially symmetric black holes.  In fact the inequality (\ref{eq:5}) was the
inspiration for the inequality (\ref{eq:22}).  A possible generalization of
(\ref{eq:5}) for bodies is to take  the area $A(\partial
\dom)$ of the boundary $\partial \dom$ of the body $\dom$
as measure of size. But
unfortunately the area of the boundary is not a good measure of the size of a
body in the presence of curvature. In particular, an inequality of the form
$A(\partial \dom) \apprge G c^{-3} |J(\dom)| $ does not hold for
bodies. The counterexample is essentially given by a rotating torus in the
weak field limit, with large major radius and small minor radius.  The details
of this calculation will be presented elsewhere \cite{Anglada13}.

It is important to emphasize that principles (i) and (iii) have a different status than principle (ii).
The former are well established facts, the later is a conjecture.
Assuming (i), (ii) and (iii) we want to argue that \eqref{eq:22} should hold.
Consider, in Newton theory, an axially symmetric  body $\dom$ with
 mass density $\bar \mu $, rotating around the axis of symmetry with angular
velocity $\omega$. These functions are not required to be constant on $\dom$.
The angular momentum and the total mass of the body are given by
\begin{equation}
  \label{eq:29}
  J(U)=  \int_\dom \bar \mu \omega \rho^2 \, \dvf, \quad  m(U)=\int_U \bar \mu \dvf,
\end{equation}
where $\rho$ is the euclidean distance to the axis and $\dvf$ is the euclidean
volume element.
The angular velocity is bounded by
\begin{equation}
  \label{eq:32}
  |\omega|=\frac{|v|}{\rho}\leq \frac{c}{\rho},
\end{equation}
where we have used the principle (i): $|v| \leq c$, where $v$ is the linear
velocity. Using \eqref{eq:32} in the expression for the angular momentum \eqref{eq:29} we obtain
\begin{equation}
  \label{eq:24}
  |J(U)|\leq c \int_\dom \bar \mu \rho \,\dvf \leq c m(U) \sup_U \rho.
\end{equation}
Note that this inequality is deduced using only the Newtonian expression for the angular momentum
and principle (i).
If the body is contained in a black hole, then the inequality  (\ref{eq:22}) holds
for the black hole boundary according to principle (iii). Hence, we assume that
it is not contained in a black hole, and then, according to principle (ii), the
inequality (\ref{eq:2}) holds. Using this inequality for the mass in
(\ref{eq:24}) we get
\begin{equation}
  \label{eq:34}
\frac{G}{c^3}  J(U)\apprle \Si(\dom) \sup_U \rho.
\end{equation}
A reasonable property for a size measure (at least in flat space) is that
\begin{equation}
  \label{eq:25}
  \sup_U \rho \leq \Si(U).
\end{equation}
Using \eqref{eq:25} in \eqref{eq:34} we obtain \eqref{eq:22}.
Note that even if the property \eqref{eq:25} does not hold,  the right hand side of \eqref{eq:34}  can be interpreted as the square of a
measure of the size of $\dom$ and hence  an inequality of the form
\eqref{eq:22} also holds for that new measure of size.

It is clear that one of the main difficulties in the study of inequalities of
the form \eqref{eq:22} and \eqref{eq:2} is the very definition of the
quantities involved, in particular the measure of size. In fact, despite the
intensive research on the subject, there is no know universal measure of size
such that the trapped surface conjecture (or, more general, the hoop
conjecture) holds (see the interesting discussions in \cite{Malec:1992ap}
\cite{Gibbons:2012ac} \cite{Senovilla:2007dw}). However, as we will see in the
next section, the remarkable point is that in order to find an appropriate
measure of size $\Si$ such that \eqref{eq:22} holds we do not to need to prove
first \eqref{eq:2}, and hence we do not need to find the relevant measure of
mass $m(\dom)$ for the trapped surface conjecture.

The arguments of the previous discussion can be summarized as follows. In order
to increase the angular momentum of a body with fixed size there are two
mechanisms: to increase the angular velocity or to increase the mass inside the
body. But there is a physical limit to both mechanisms. The angular velocity is
bounded by the speed of light, and increasing the mass (at fixed size) will
eventually produce a black hole, where the inequality \eqref{eq:22}
holds. Hence, an universal inequality of the form \eqref{eq:22} is expected for
all bodies.

\emph{A precise version of the inequality} ---  We make precise the three notions
involved in the inequality (\ref{eq:22}): a body $\dom$, the angular momentum
$J$ and the size $\Si$ of the body.  A \emph{body} $\dom$ is a connected open
subset $\dom \subset S$ with smooth boundary $\partial \dom$; where $S$ is
a spacelike 3-surface which gives rise to the initial data set for Einstein
equations defined as follows. An \emph{initial data set} for the Einstein
equations is given by $(S, h_{ij}, K_{ij},\mu, j^i)$ where $S$ is a connected
3-dimensional manifold, $ h_{ij} $ a (positive definite) Riemannian metric,
$ K_{ij}$ a symmetric tensor field, $j^i$  a vector field and $\mu$ a scalar
field on $S$, such that the constraint equations
\begin{align}
 \label{const1}
   D_j   K^{ij} -  D^i   K= -8\pi\frac{G}{c^4} j^i,\\
 \label{const2}
   R -  K_{ij}   K^{ij}+  K^2=16\pi\frac{G}{c^4} \mu,
\end{align}
are satisfied on $S$. Where $ {D}$ and $R$ are the Levi-Civita connection and
the scalar curvature  associated with $ {h}_{ij}$, and $ K = K_{ij} h^{ij}$. In
these equations the indices are moved with the metric $ h_{ij}$ and its inverse
$ h^{ij}$. In terms of the four dimensional energy momentum tensor $T_{\mu\nu}$,
the matter fields are given by $\mu=T_{\mu\nu} n^\mu n^\nu$, $j_\nu=-
h_\nu{}^\lambda T_{\lambda \nu} n^\nu$,
where $n^\nu$ is the timelike unit vector normal to the slice $S$. The relation
between the mass
density $\bar \mu$ used in (\ref{eq:29}) and the energy density
$\mu$ is given  by $\mu=c^2 \bar \mu$.

We require that the matter fields satisfy the dominant energy condition
\begin{equation}
  \label{eq:1}
  \mu\geq \sqrt{j^i j_i}.
\end{equation}
In order to have a proper definition of the angular momentum of the body we
will further assume that the data are \emph{axially symmetric} (in general, the
angular momentum of a bounded region $\dom$ is very difficult to define, see
the review article \cite{Szabados04} and reference therein). That is, we assume
the existence of a Killing vector field $\eta^i$, i.e;
\begin{equation}
  \label{eq:38}
 \pounds_\eta   h_{ij}=0,
\end{equation}
where $\pounds$ denotes the Lie derivative, which has complete periodic orbits
and such that
\begin{equation}
  \label{eq:8b}
  \pounds_\eta \mu  = \pounds_\eta j^j = \pounds_\eta   K_{ij}=0.
 \end{equation}
We denote the  norm of the Killing vector by $\norm= (\eta^i\eta_i)^{1/2}$.
The angular momentum of the body $\dom$ is defined by
\begin{equation}
  \label{eq:4}
  J(\dom)=-\frac{1}{c}\int_{\dom} j_i \eta^i \dv,
\end{equation}
where $\dv$ is the volume measure with respect to the metric $h_{ij}$.

Finally, we should define a notion of size for the body $\dom$. This notion
will be a variant of the following definition of radius presented by Schoen and Yau in
\cite{schoen83d}. Let $\Gamma$ be a simple closed curve
in $\dom$ which bounds a disk in $\dom$. Let $p$ be largest constant such
that the set of points within a distance $p$ of $\Gamma$ is contained within
$\dom$ and forms a proper torus. Then $p$ is a measure of the size of
$\dom$ with respect to the curve $\Gamma$. The radius $\ra(\dom)$ is
defined as the largest value of $p$ we can find by considering all curves
$\Gamma$. That is, $\ra(\dom)$ is expressed in terms of the largest torus
that can be embedded in $\dom$.  Using this definition, the following deep
theorem was proved in \cite{schoen83d}. Let $\dom$ be any subset of $S$.
Assume that the scalar curvature $R$ of the metric $h_{ij}$ is bounded from
below $R\geq \Lambda$ in $\dom$ by a positive constant $\Lambda$. Then the
following inequality holds
\begin{equation}
  \label{eq:49}
  \Lambda \leq  \frac{8\pi^2}{3}\frac{1}{ \ra^{2}}.
\end{equation}
Note that this is a purely local and purely Riemannian result.  There is no
requirement that $S$ be asymptotically flat and only assumptions on the metric
$h_{ij}$ are made.

In \cite{Murchadha86b} \'O Murchadha made the following important
observation. Define another radius $\ram(\dom)$ as follows. Let
$\ram(\dom)$ be the size of the largest stable minimal 2-surface that can be
imbedded in $\dom$, where size of the surface is the distance (with respect
to the ambient metric $h_{ij}$) from the boundary to that internal point which
is furthest from the boundary. Then, it can be proved that
\begin{equation}
  \label{eq:9}
  \ram(\dom) \geq \ra (\dom),
\end{equation}
and also that the same bound (\ref{eq:49}) holds for $\ram(\dom)$, under
similar assumptions\footnote{It was pointed out in \cite{Galloway:2008gc} that
  to prove this bound with the radius $\ram$ an additional requirement is
  needed:  the boundary $\partial \dom$ should be mean convex.}. Namely,
\begin{equation}
  \label{eq:49b}
  \Lambda \leq  \frac{8\pi^2}{3}\frac{1}{ \ram^{2}}.
\end{equation}
Since we have \eqref{eq:9}, the right hand side of \eqref{eq:49b} is smaller
than the right hand side of \eqref{eq:49}, and hence $\ram$ provide a better
bound.

To have an intuitive idea of these measures, let us compute them for some
relevant domains in flat space. Recall that the planes are minimal stable
surfaces in flat space. For a sphere of radius $b$ we have that
$\ra=b/2$, $\ram=b$.  We see that both radii give essentially the same desired
value for the sphere.  For a torus with major radius $b$ and minor radius $a$
we have $\ra=a/2$, $\ram=a$.  Both radii are independent of the major radius
$b$ for the torus. Hence, we can not expect an inequality of the form
\eqref{eq:22} for $\ra$ or $\ram$, since in the weak field limit a torus of
large radius $b$ and small radius $a$ will have large angular momentum $J$ and
small $\ra$ or $\ram$ (a similar counter example as in the case of the area
discussed above).  Finally, to see the relevant difference between $\ra$ and
$\ram$ consider a cylinder with radius $a$ and height $L$.  We have
$\ra=\min\{a/2,L/2\}$, $\ram=a$. When $L>a$, then both radius gives similar
values, however for a thin disk with $L<a$ we have $\ra=L/2$ and $\ram=a$.
That is, $\ra\to 0$ as $L \to 0$ while $\ram$ is independent of $L$.

Motivated by the example of the torus, we define a new radius for axially
symmetric bodies as follows. Consider a region $\dom$ with a Killing vector
$\eta^i$ with norm $\lambda$,  we define the radius $\Si$ by
\begin{equation}
  \label{eq:8}
  \Si(\dom) = \frac{2}{\pi}\frac{\left(\int_\dom \norm \, \dv\right)^{1/2}}{\ram(\dom)}.
\end{equation}
This will be our measure for size for the inequality \eqref{eq:22}.
The most natural normalization for $\Si$ in the inequality \eqref{eq:22}
is to require that $\Si=b$ for an sphere in flat space of radius $b$. This is
the reason for  the factor $2/\pi$ in \eqref{eq:8}.

We have also the analog definition with respect to $\ra$, namely
\begin{equation}
  \label{eq:15}
  \Sip(\dom) = \frac{2}{\pi} \frac{\left(\int_\dom \norm \,
      \dv\right)^{1/2}}{\ra(\dom)}.
\end{equation}
Using the inequality \eqref{eq:9}, we obtain
\begin{equation}
  \label{eq:16}
  \Sip(\dom) \geq \Si(\dom).
\end{equation}
That is, from the point of view of the inequality \eqref{eq:22}, the radius
$\Si$ provides a sharper estimate than $\Sip$.

For the torus in flat space, the volume integral of the norm of the Killing
vector is given by
\begin{equation}
  \label{eq:19}
  \int_{\text{Torus}} \rho \, \dvf=2\pi^2 a^2 \left(\frac{a^2}{4}+b^2\right).
\end{equation}
Then we obtain
\begin{equation}
  \label{eq:20}
  \Si=  2^{3/2}\left(\frac{a^2}{4}+b^2\right)^{1/2}, \quad   \Sip= 2\Si.
\end{equation}
The important
point is that in the limit $a\to 0$ we obtain $\Si=2^{3/2}b$, that is, a
torus with a large $b$ has also large size in contrast with the original radii
$\ra$ or $\ram$.
For a thin disk with $L<a$ we have
\begin{equation}
  \label{eq:21}
   \Si=\frac{2^{3/2}}{\sqrt{3\pi}} \sqrt{aL},  \quad \Sip=\frac{2^{5/2}}{\sqrt{3\pi}}
   \frac{a^{3/2}}{L^{1/2}}.
\end{equation}
We see that $\Si\to 0$ and  $\Sip\to \infty$ as $L \to 0$. That is, the
difference between the two measures is significant.

Finally, it is important to compute $\Si$ for a very dense body where the
gravitational field is strong.  Consider a constant density star of total mass
$m$ with area radius equal to Schwarzschild radius $2m G/c^2$. That is, we are
considering the limit case before the formation of a black hole. The radius
$\ram$ for that case was calculated in  \cite{Murchadha86b}. Using that result
we obtain
\begin{equation}
  \label{eq:13}
  \Si=\frac{2^{11/2}}{\pi\sqrt{3}} \frac{G}{c^2} m \approx 8.16  \frac{G}{c^2} m.
\end{equation}
We see that $\Si$ is of the same order of magnitude than the area radius, and
hence it is a reasonable measure of size in that case.

We have the following result.
\begin{theorem}
\label{t:1}
  Let $(S,h_{ij},K_{ij},\mu, j^i)$ be an initial data set that satisfy the energy
  condition \eqref{eq:1}. We assume that the data are maximal (i.e. $K=0$) and
  axially symmetric.  Let
  $\dom$ be an open set in $S$. Assume that
the energy density  $\mu $ is constant on $\dom$.  Then the following
 inequality holds
 \begin{equation}
   \label{eq:7}
 \Sip^2(\dom) \geq  \frac{24}{\pi^3}\frac{G}{c^3} |J(\dom)|.
 \end{equation}
The same bound holds for $\Si(\dom)$ if we assume, in addition, that the
boundary $\partial \dom$ is mean convex.
\end{theorem}

\begin{proof}
The angular momentum of the body $\dom$ is given by
(\ref{eq:4}). Define the unit vector $\hat \eta^i$ by
\begin{equation}
  \label{eq:14}
  \hat \eta ^i=\frac{\eta^i}{\norm}.
\end{equation}
Then we have
\begin{align}
  |J(\dom)|\leq \frac{1}{c} \int_\dom |j^i\eta_i| \, dv &=\frac{1}{c} \int_\dom
  |j^i\hat\eta_i|\,
  \norm \, dv \label{eq:bj1} \\
&\leq \frac{1}{c} \int_\dom \sqrt{j^ij_i}\, \norm\, dv  \label{eq:bj2}\\
&\leq \frac{1}{c}\int_\dom \mu \, \norm \, dv  ,\label{eq:bj3}
\end{align}
where in the line \eqref{eq:bj2} we have used that $\hat \eta^i$ has unit norm,
in the line \eqref{eq:bj3} we used the energy condition (\ref{eq:1}).

We have assumed that the data are maximal and hence, by equation
(\ref{const2}), we obtain
\begin{equation}
  \label{eq:23}
  R\geq 16\pi\frac{G}{c^4} \mu.
\end{equation}
Since we have assumed that $\mu$ is constant (which should be positive by the
energy condition  (\ref{eq:1})) on $\dom$, we can take
$\Lambda=16\pi G c^{-4}\mu$ and then we are under the hypothesis of the Schoen-Yau theorem.
That is, the bound (\ref{eq:49}) holds, and hence we get
\begin{equation}
  \label{eq:17}
  \mu \leq \frac{\pi}{6}\frac{c^4}{G}\frac{1}{ \ra^{2}}.
\end{equation}
Using this bound in (\ref{eq:bj3}) we obtain
\begin{equation}
  \label{eq:6}
  |J(\dom)| \leq \frac{\pi}{6} \frac{c^3}{G}\frac{1}{ \ram^{2}}\int_\dom \norm \, \dv
  = \frac{\pi^3}{24} \frac{c^3}{G} \Sip^2,
\end{equation}
where in the last equality we have used the definition (\ref{eq:15}). Under the
additional assumption that the boundary $\partial \dom$ is mean convex, we
have the same bound (\ref{eq:17}) for the radius $\ram$, and hence the same
inequality (\ref{eq:6}) holds for $\Si$.
\end{proof}

It is interesting to note that this proof is very similar to the heuristic
argument presented above. There is a physical reason for this similarity: in
axial symmetry  the gravitational waves have no angular
momentum. All the angular momentum is contained in the matter sources. Hence
the Newtonian expression for the angular momentum \eqref{eq:29} is similar to
relativistic one \eqref{eq:4}. Condition (i) on the maximum velocity of the
matter is expressed in the dominant energy condition \eqref{eq:1}. Moreover,
from inequality (\ref{eq:bj3}) (without using the assumption that $\mu$ is
constant), we get the analog of the inequality (\ref{eq:24}), namely
\begin{equation}
  \label{eq:3}
  |J(U)|\leq c m(U) \sup_U \lambda,
\end{equation}
where we have defined
\begin{equation}
  \label{eq:11}
  m(U)=\frac{1}{c^2}\int_U \mu \dv.
\end{equation}
Note that the length of the azimuthal circles is given by $2\pi \lambda$,
hence $\lambda$ represents a natural generalization for curved spaces of the
coordinate $\rho$ that appears in (\ref{eq:24}).

The important new ingredient is that instead of using the bound \eqref{eq:2}
for the mass of the body, we use the Schoen-Yau bound for the energy density
\eqref{eq:49}. This allow us to bypass the hoop conjecture and its associated
definition of size and mass.

Note that the radius used in the theorem can not be
applied in general to black holes, since it requires a regular interior region.
And even when the interior is regular the radius is not a priori related with the black hole area.
A relevant open problem is to find a suitable measure of size that can be applied for both
 black holes and bodies.

\emph{Physical relevance} ---
It is important to emphasize that the validity of inequality \eqref{eq:22} is
entirely independent of any specific matter model, the only requirement is that
the dominant energy condition is satisfied.

The inequality \eqref{eq:22} is a prediction of Einstein theory and hence it
should be contrasted with observational evidences. In order to violate this
inequality a body should be small and highly spinning, a natural candidate for
that is a neutron star. For the fastest rotating neutron star found to date
(see \cite{Hessels:2006ze}) we have
\begin{equation}
  \label{eq:66}
  \omega \approx 4.5 \times 10^3 \, rad\, s^{-1}.
\end{equation}
Assuming that the neutron star has about three solar masses (which appears to
be a reasonable upper bound for the mass, see \cite{Lattimer:2004pg}) we obtain
\begin{equation}
  \label{eq:67}
  m \omega\approx 2.7 \times 10^{37}\, s^{-1} g.
\end{equation}
The  radius of the neutron star is typically
\begin{equation}
  \label{eq:12}
  \Si \approx 1.2 \times 10^6 \, cm.
\end{equation}
Assuming that the star is spherical with constant density we get that the
angular momentum is given by
\begin{equation}
  \label{eq:10}
 \frac{G}{c^3} |J| = \frac{G}{c^3}\frac{2}{5}m \Si^2 \omega \approx 3.8 \times
 10^{10} \, cm^2.
\end{equation}
This should be compared with the square of the  radius
\begin{equation}
  \label{eq:18}
  \Si^2 \approx 1.44 \times 10^{12} \, cm^2.
\end{equation}
We see that the inequality (\ref{eq:22}) is satisfied.

Finally, it is also interesting to consider what kind of limit the inequality
\eqref{eq:22} impose on elementary particles.  From quantum mechanics we get that  the
angular momentum of an elementary particle is given by
\begin{equation}
  \label{eq:13b}
  J=\sqrt{s(s+1)}\hbar, \quad   \hbar  =1.05 \times 10^{-27} \, cm^2  s^{-1} g,
\end{equation}
where $s$ is the spin of the particle. Using this expression in \eqref{eq:22}
we obtain that the classical theory impose the following  minimal size
for a particle with spin $s$
\begin{equation}
  \label{eq:22b}
 \Si_0 = (s(s+1))^{1/4} l_p, \quad l_p = \left( \frac{G \hbar}{c^3}\right)^{1/2},
\end{equation}
where $l_p=1.6\times 10^{-33}\, cm$ is the Planck length.  We recover the
Planck length essentially because the order of magnitude of the universal
constant in the inequality \eqref{eq:22} is one. It appears to be a remarkable
self consistence of the Einstein field equations that they predict a minimum
length of the order of magnitude of the Planck length if we assume that there
exists a minimum for the angular momentum given by quantum mechanics.

\begin{acknowledgments}
  It is a pleasure to thank E. Gallo, G. Galloway, R. J. Gleiser, N. \'O
  Murchadha, O. Ortiz, M. Reiris, R. Wald, for illuminating discussions.  This
  work was supported in by grant PICT-2010-1387 of CONICET (Argentina) and
  grant Secyt-UNC (Argentina).
\end{acknowledgments}


\end{document}